\documentclass[twoside,onecolumn,a4paper,final]{article}%
\usepackage{amssymb}
\usepackage{amsfonts}
\usepackage{amsmath}
\usepackage{graphicx}%
\providecommand{\U}[1]{\protect\rule{.1in}{.1in}}
\newtheorem{theorem}{Theorem}

\newtheorem{corollary}[theorem]{Corollary}

\newtheorem{lemma}[theorem]{Lemma}

\newtheorem{remark}[theorem]{Remark}

\newenvironment{proof}[1][Proof]{\noindent\textbf{#1.} }{\ \rule{0.5em}{0.5em}}
\begin{document}

\title{\textbf{Vortical and Self-similar Flows of 2D Compressible Euler Equations}}
\author{M\textsc{anwai Yuen\thanks{E-mail address: nevetsyuen@hotmail.com }}\\\textit{Department of Mathematics and Information Technology,}\\\textit{The Hong Kong Institute of Education,}\\\textit{10 Po Ling Road, Tai Po, New Territories, Hong Kong}}
\date{Revised 19-Jan-2013}
\maketitle

\begin{abstract}
This paper presents the vortical and self-similar solutions for 2D
compressible Euler equations using the separation method. These solutions
complement Makino's solutions in radial symmetry without rotation. The
rotational solutions provide new information that furthers our understanding
of ocean vortices and reference examples for numerical methods. In addition,
the corresponding blowup, time-periodic or global existence conditions are
classified through an analysis of the new Emden equation. A conjecture
regarding rotational solutions in 3D is also made.

\ 


\ 


\end{abstract}

\section{Introduction}

In fluid mechanics, the N-dimensional isentropic compressible Euler and
Navier-Stokes equations are expressed as follows:%
\begin{equation}
\left\{
\begin{array}
[c]{rl}%
{\normalsize \rho}_{t}{\normalsize +\nabla\cdot(\rho\vec{u})} &
{\normalsize =}{\normalsize 0}\\[0.08in]%
\rho\lbrack\vec{u}_{t}+(\vec{u}\cdot\nabla)\vec{u}]+K\nabla\rho^{\gamma} & =0,
\end{array}
\right.  \label{e1}%
\end{equation}
where $\rho=\rho(t,\vec{x})$ denotes the density of the fluid, $\vec{u}%
=\vec{u}(t,\vec{x})=(u_{1},u_{2},\cdots,u_{N})\in R^{N}$ is the velocity,
$\vec{x}=(x_{1},x_{2},\cdots,x_{N})\in R^{N}$ is the space variable, and
$K>0,\;\gamma>1$ are constants.

The Euler equations have been studied in great detail by numerous scholars
because of their significance in a variety of physical fields, such as fluids,
plasmas, condensed matter, astrophysics, oceanography, and atmospheric
dynamics. These equations are also important in physics and are widely used in
different areas of study. For instance, the Euler system is the basic model
for shallow water flows \cite{ConstantinA}. It also provides a good model of
the superfluids produced by the Bose-Einstein condensates in the dilute gases
of alkali metal atoms, in which identical gases do not interact at very low
temperatures \cite{Einzel}. At the microscopic level, fluids or gases are
formed by many tiny discrete molecules or particles that collide with one
another. As the cost of directly calculating the particle-to-particle or
molecule-to-molecule evolution of fluids on a large scale is expensive,
approximation methods are needed to considerably simplify the process.
Therefore, at the macroscopic scale, the continuum assumption, which considers
fluids as continuous, is applied to the modeling. Here, the Euler equations
provide a good description of the fluids at the statistical limit of a large
number of small ideal molecules or particles by ignoring the less influential
effects, such as the self-gravitational forces and relativistic effect
\cite{CIP}. For a mathematical introduction of the Euler equations, readers
are referred to \cite{Lion} and \cite{CW}.

The construction of analytical or exact solutions is an important area in
mathematical physics and applied mathematics, as it can further classify their
nonlinear phenomena. For non-rotational flows, Makino first obtained the
radial symmetry solutions for the Euler equations (\ref{e1}) in $R^{N}$ in
1993 \cite{Makino93exactsolutions}. A number of special solutions for these
equations \cite{Li}, \cite{LW}, \cite{Yuen3DexactEuler}, \cite{YuenPLA}, and
\cite{YuenCNSNS2012} were subsequently obtained. For rotational flows, Zhang
and Zheng \cite{ZZ} constructed explicitly rotational solutions for the Euler
equations with $\gamma=2$ in 1997:%
\begin{equation}
\rho=\frac{r^{2}}{8Kt^{2}},\text{ }u_{1}=\frac{1}{2t}(x+y),\text{ }u_{2}%
=\frac{1}{2t}(x-y), \label{ZZZ}%
\end{equation}
where $x=r\cos\theta$ and $y=r\sin\theta.$

Very recently, Kwong and Yuen \cite{KYIso} constructed a family of rotational
solutions for the Euler-Poisson equations
\begin{equation}
\left\{
\begin{array}
[c]{rl}%
{\normalsize \rho}_{t}{\normalsize +\nabla\cdot(\rho\vec{u})} &
{\normalsize =}{\normalsize 0}\\
\rho(\vec{u}_{t}+(\vec{u}\cdot\nabla)\vec{u}){\normalsize +K\nabla\rho} &
{\normalsize =}{\normalsize -\rho\nabla\Phi}\\
{\normalsize \Delta\Phi(t,\vec{x})} & {\normalsize =2\pi}{\normalsize \rho,}%
\end{array}
\right.  \label{Euler-Poisson2DRotation}%
\end{equation}
with $N=2$ and $\gamma=1$:\newline%
\begin{equation}
\left\{
\begin{array}
[c]{c}%
\rho(t,\vec{x})=\frac{1}{a(t)^{2}}e^{f(r/a(t))}\text{, }{\normalsize \vec
{u}(t,\vec{x})=}\frac{\overset{\cdot}{a}(t)}{a(t)}(x,y){+}\frac{\xi}{a(t)^{2}%
}(-y,x)\\
\ddot{a}(t)=\frac{-\lambda}{a(t)}+\frac{\xi^{2}}{a(t)^{3}}\text{, }%
a(0)=a_{0}>0\text{, }\dot{a}(0)=a_{1}\\
\overset{\cdot\cdot}{f}(s){\normalsize +}\frac{1}{s}\overset{\cdot}%
{f}(s){\normalsize +\frac{2\pi}{K}e}^{f(s)}{\normalsize =}\frac{2\lambda}%
{K}{\normalsize ,}\text{ }f(0)=\alpha,\text{ }\overset{\cdot}{f}(0)=0,
\end{array}
\right.  \label{2-DIsothermalRotation}%
\end{equation}
with arbitrary constants $\xi\neq0,$ $a_{0}$, $a_{1}$ and $\alpha$.\newline
The rotational case complements Yuen's solutions without rotation ($\xi=0$)
\cite{YuenJMAA2008a}. In this paper, based on the foregoing development, we
provide the corresponding vortical flows for 2D compressible Euler equations
(\ref{e1}) with $\gamma>1$ in the following result.

\begin{theorem}
\label{thm:1}For $\gamma>1$, there exists a family of vortical flows in radial
symmetry for the compressible Euler equations (\ref{e1}) in 2D,%
\begin{equation}
\left\{
\begin{array}
[c]{c}%
\rho(t,\vec{x})=\frac{\max\left(  \left(  -\frac{\lambda(\gamma-1)}{2K\gamma
}s+\alpha\right)  ^{\frac{1}{\gamma-1}},\text{ }0\right)  }{a(t)^{2}}\\
{\normalsize \vec{u}(t,\vec{x})=}\frac{\overset{\cdot}{a}(t)}{a(t)}%
(x,y){+}\frac{\xi}{a(t)^{2}}(-y,x)\\
\ddot{a}(t)-\frac{\xi^{2}}{a(t)^{3}}-\frac{\lambda}{a(t)^{2\gamma-1}}=0\text{,
}a(0)=a_{0}>0\text{, }\dot{a}(0)=a_{1},
\end{array}
\right.  \label{2-Dg>1Rotation}%
\end{equation}
with the self-similar variable $s=\frac{x^{2}+y^{2}}{a(t)^{2}}$ and arbitrary
constants $\xi\neq0,$ $a_{0}$, $a_{1},$ and $\alpha$.
\end{theorem}

\begin{remark}
This result complements Makino's solutions in radial symmetry without rotation
($\xi=0$). The vortical solutions (\ref{2-Dg>1Rotation}) provide new
information that may further our understanding of oceans vortices and
reference examples for numerical methods in computational physics.
\end{remark}

\begin{remark}
Zhang and Zheng's solution (\ref{e1}) for $\gamma=2$ is a special case in our
solutions (\ref{2-Dg>1Rotation}).
\end{remark}

\section{\textbf{2D Vortical and Self-similar Flows}}

Here, we provide a lemma for the vortical flows in 2D for the mass equation
(\ref{e1})$_{1}$. This lemma originates in Kwong and Yuen's paper
\cite{KYIso}, which constructs periodic and spiral solutions for 2D
Euler-Poisson equations (\ref{Euler-Poisson2DRotation}) with $\gamma=1$.

\begin{lemma}
[\cite{KYIso}]\label{lem:generalsolutionformasseqrotation2d}For the equation
of the conservation of mass,
\begin{equation}
\rho_{t}+\nabla\cdot\left(  \rho\vec{u}\right)  =0,
\label{massequationspherical2Drotation}%
\end{equation}
there exist the following solutions:%
\begin{equation}
\rho(t,\vec{x})=\rho(t,r)=\frac{f(\frac{r}{a(t)})}{a(t)^{2}},\text{ }\vec
{u}(t,\vec{x})=\frac{\dot{a}(t)}{a(t)}\left(  x,y\right)  +\frac{G(t,r)}%
{r}(-y,x) \label{generalsolutionformassequation2Drotation}%
\end{equation}
with the radial $r=\sqrt{x^{2}+y^{2}}$ and arbitrary functions $f\geq0$; $G$
and $a(t)>0\in C^{1}.$

\begin{proof}
The functional structure%
\begin{equation}
\rho(t,\vec{x})=\rho(t,r),\text{ }\vec{u}(t,\vec{x})=\frac{F(t,r)}{r}\left(
x,y\right)  +\frac{G(t,r)}{r}(-y,x)
\end{equation}
with an arbitrary $C^{1}$ function $F(t,r),$ can be plugged into the mass
equation (\ref{massequationspherical2Drotation}) to verify the result:
\begin{equation}
\rho_{t}+\nabla\cdot\left(  \rho\vec{u}\right)
\end{equation}%
\begin{equation}
=\rho_{t}+\frac{\partial}{\partial x}\left[  \rho\frac{Fx}{r}-\rho\frac{Gy}%
{r}\right]  +\frac{\partial}{\partial y}\left[  \rho\frac{Fy}{r}+\rho\frac
{Gx}{r}\right]
\end{equation}%
\begin{align}
&  =\rho_{t}+\frac{\partial}{\partial x}\rho\frac{Fx}{r}+\rho\left(
\frac{\partial}{\partial x}\frac{Fx}{r}\right)  -(\frac{\partial}{\partial
x}\rho)\frac{Gy}{r}-\rho(\frac{\partial}{\partial x}\frac{Gy}{r})\nonumber\\
&  +\frac{\partial}{\partial y}\rho\frac{Fy}{r}+\rho\left(  \frac{\partial
}{\partial y}\frac{Fy}{r}\right)  +(\frac{\partial}{\partial y}\rho)\frac
{Gx}{r}+\rho(\frac{\partial}{\partial y}\frac{Gx}{r})
\end{align}%
\begin{align}
&  =\rho_{t}+\rho_{r}\frac{x}{r}\frac{Fx}{r}+\rho\left(  F_{r}\frac{x}%
{r}\right)  \frac{x}{r}+\rho\frac{F}{r}-\rho Fx\frac{x}{r^{3}}\nonumber\\
&  -\rho_{r}\frac{x}{r}\frac{Gy}{r}-\rho G_{r}\frac{x}{r}\frac{y}{r}+\rho
Gy\frac{x}{r^{3}}+\rho_{r}\frac{y}{r}\frac{Fy}{r}+\rho\left(  F_{r}\frac{y}%
{r}\right)  \frac{y}{r}\nonumber\\
&  +\rho\frac{F}{r}-\rho Fy\frac{y}{r^{3}}+\rho_{r}\frac{y}{r}\frac{Gx}%
{r}+\rho\left(  G_{r}\frac{y}{r}\right)  \frac{x}{r}-\rho Gx\frac{y}{r^{3}}%
\end{align}%
\begin{align}
&  =\rho_{t}+\rho_{r}\frac{x}{r}\frac{Fx}{r}+\rho\left(  F_{r}\frac{x}%
{r}\right)  \frac{x}{r}+\rho\frac{F}{r}-\rho Fx\frac{x}{r^{3}}\nonumber\\
&  +\rho_{r}\frac{y}{r}\frac{Fy}{r}+\rho\left(  F_{r}\frac{y}{r}\right)
\frac{y}{r}+\rho\frac{F}{r}-\rho Fy\frac{y}{r^{3}}%
\end{align}%
\begin{equation}
=\rho_{t}+\rho_{r}F+\rho F_{r}+\rho F\frac{1}{r}. \label{tomassradial}%
\end{equation}
Then, the self-similar structure is taken for the density function,%
\begin{equation}
\rho(t,\vec{x})=\rho(t,r)=\frac{f(\frac{r}{a(t)})}{a(t)^{2}}%
\end{equation}
and $F(t,r)=\frac{\dot{a}(t)}{a(t)}r$ for velocity $\vec{u}$ to balance
equation (\ref{tomassradial}) \cite{YuenJMAA2008a}:%
\begin{equation}
=\frac{\partial}{\partial t}\frac{f(\frac{r}{a(t)})}{a(t)^{2}}+\frac{\partial
}{\partial r}\frac{f(\frac{r}{a(t)})}{a(t)^{2}}\frac{\dot{a}(t)r}{a(t)}%
+\frac{f(\frac{r}{a(t)})}{a(t)^{2}}\frac{\dot{a}(t)}{a(t)}+\frac{f(\frac
{r}{a(t)})}{a(t)^{2}}\frac{\dot{a}(t)}{a(t)}%
\end{equation}%
\begin{align}
&  =\frac{-2\overset{\cdot}{a}(t)f(r/a(t))}{a(t)^{3}}-\frac{\overset{\cdot}%
{a}(t)r\overset{\cdot}{f}(r/a(t))}{a(t)^{4}}\nonumber\\
&  +\frac{\overset{\cdot}{f}(r/a(t))}{a(t)^{3}}\frac{\overset{\cdot}{a}%
(t)r}{a(t)}+\frac{f(r/a(t))}{a(t)^{2}}\frac{\overset{\cdot}{a}(t)}{a(t)}%
+\frac{f(r/a(t))}{a(t)^{2}}\frac{\overset{\cdot}{a}(t)}{a(t)}%
\end{align}%
\begin{equation}
=0.
\end{equation}
The proof is completed.
\end{proof}
\end{lemma}

The computational proof for Theorem 1 is as follows.

\begin{proof}
[Proof of Theorem 1]The procedure for the proof for vortical fluids is similar
to that for non-vortical fluids \cite{Makino93exactsolutions},
\cite{YuenCNSNS2012}. It is clear that the following function
\begin{equation}
\rho(t,\vec{x})=\frac{f(s)}{a(t)^{2}},\text{ }{\normalsize \vec{u}(t,\vec
{x})=}\frac{\overset{\cdot}{a}(t)}{a(t)}(x,y){+}\frac{\xi}{a(t)^{2}}(-y,x),
\end{equation}
with an arbitrary $C^{1}$ function $f\geq0$ of the self-similar variable
$s=\frac{x^{2}+y^{2}}{a(t)^{2}}=\frac{r^{2}}{a(t)^{2}}$ and a undetermined
$C^{2}$ time-function $a(t)>0$, satisfies Lemma
\ref{lem:generalsolutionformasseqrotation2d} for the mass equation
(\ref{e1})$_{1}$. For the first momentum equation (\ref{e1})$_{21}$, we
obtain:%
\begin{equation}
=\rho\left(  \frac{\partial}{\partial t}u_{1}+u_{1}\frac{\partial u_{1}%
}{\partial x}+u_{2}\frac{\partial u_{1}}{\partial y}\right)  +K\gamma
\rho^{\gamma-1}\frac{\partial}{\partial x}\rho
\end{equation}%
\begin{equation}
=\rho\left(  \frac{\partial}{\partial t}u_{1}+u_{1}\frac{\partial u_{1}%
}{\partial x}+u_{2}\frac{\partial u_{1}}{\partial y}\right)  +K\gamma
\frac{f\left(  s\right)  ^{\gamma-1}}{a(t)^{2(\gamma-1)}}\frac{1}{a(t)^{2}%
}\frac{\partial}{\partial x}f(s)
\end{equation}%
\begin{equation}
=\rho\left[  \frac{\partial}{\partial t}u_{1}+u_{1}\frac{\partial u_{1}%
}{\partial x}+u_{2}\frac{\partial u_{1}}{\partial y}+K\gamma\frac
{f(s)^{\gamma-2}}{a(t)^{2(\gamma-1)}}\frac{2x}{a(t)^{2}}\dot{f}(s)\right]
\end{equation}%
\begin{equation}
=\rho\left[
\begin{array}
[c]{c}%
\frac{\partial}{\partial t}\left(  \frac{\dot{a}(t)}{a(t)}x-\frac{\xi
}{a(t)^{2}}y\right)  +\left(  \frac{\dot{a}(t)}{a(t)}x-\frac{\xi}{a(t)^{2}%
}y\right)  \frac{\partial}{\partial x}\left(  \frac{\dot{a}(t)}{a(t)}%
x-\frac{\xi}{a(t)^{2}}y\right) \\
+\left(  \frac{\xi}{a(t)^{2}}x+\frac{\dot{a}(t)}{a(t)}y\right)  \frac
{\partial}{\partial y}\left(  \frac{\dot{a}(t)}{a(t)}x-\frac{\xi}{a(t)^{2}%
}y\right)  +K\gamma\frac{f(s)^{\gamma-2}}{a(t)^{2\gamma}}2x\dot{f}(s)
\end{array}
\right]
\end{equation}%
\begin{equation}
=\rho\left[
\begin{array}
[c]{c}%
\left(  \frac{\ddot{a}(t)}{a(t)}-\frac{\dot{a}(t)^{2}}{a(t)^{2}}\right)
x+\frac{2\xi\dot{a}(t)}{a(t)^{3}}y+\left(  \frac{\dot{a}(t)}{a(t)}x-\frac{\xi
}{a(t)^{2}}y\right)  \frac{\dot{a}(t)}{a(t)}\\
+\left(  \frac{\xi}{a(t)^{2}}x+\frac{\dot{a}(t)}{a(t)}y\right)  \left(
-\frac{\xi}{a(t)^{2}}\right)  +K\gamma\frac{f(s)^{\gamma-2}}{a(t)^{2\gamma}%
}2x\dot{f}(s)
\end{array}
\right]
\end{equation}%
\begin{equation}
=\rho\left[  \left(  \frac{\ddot{a}(t)}{a(t)}-\frac{\xi^{2}}{a(t)^{4}}\right)
x+K\gamma\frac{f(s)^{\gamma-2}}{a(t)^{2\gamma}}2x\dot{f}(s)\right]
\end{equation}%
\begin{equation}
=\frac{\rho x}{a(t)^{2\gamma}}\left[  \lambda+2K\gamma f(s)^{\gamma-2}\dot
{f}(s)\right]  , \label{eq58}%
\end{equation}
with the Emden equation%
\begin{equation}
\left\{
\begin{array}
[c]{c}%
\ddot{a}(t)-\frac{\xi^{2}}{a(t)^{3}}=\frac{\lambda}{a(t)^{2\gamma-1}}\\
a(0)=a_{0}>0\text{, }\dot{a}(0)=a_{1},
\end{array}
\right.  \label{Endemeqeq2-Drotation}%
\end{equation}
with arbitrary constants $\xi$ and $\lambda>1$.\newline To ensure that
equation (\ref{eq58}) is zero, the following ordinary differential equation is
required:%
\begin{equation}
\left\{
\begin{array}
[c]{c}%
\lambda+2K\gamma f(s)^{\gamma-2}\dot{f}(s)=0\\
f(0)=\alpha>0\text{, }\dot{f}(0)=0.
\end{array}
\right.  \label{Liouville2-DRotation}%
\end{equation}
The function $f(s)$ in the foregoing ordinary differential equation can be
solved exactly:%
\begin{equation}
f(s)=\left(  -\frac{\lambda(\gamma-1)}{2K\gamma}s+\alpha\right)  ^{\frac
{1}{\gamma-1}}. \label{fexact}%
\end{equation}
\newline For the second momentum equation (\ref{e1})$_{22}$, we also have:%
\begin{equation}
=\rho\left(  \frac{\partial}{\partial t}u_{2}+u_{1}\frac{\partial u_{2}%
}{\partial x}+u_{2}\frac{\partial u_{2}}{\partial y}\right)  +K\gamma
\rho^{\gamma-1}\frac{\partial}{\partial y}\rho
\end{equation}%
\begin{equation}
=\rho\left[  \frac{\partial}{\partial t}u_{2}+u_{1}\frac{\partial u_{2}%
}{\partial x}+u_{2}\frac{\partial u_{2}}{\partial y}+K\gamma\frac
{f(s)^{\gamma-2}}{a(t)^{2(\gamma-1)}}\frac{2y}{a(t)^{2}}\dot{f}(s)\right]
\end{equation}%
\begin{equation}
=\rho\left[
\begin{array}
[c]{c}%
\frac{\partial}{\partial t}\left(  \frac{\xi}{a(t)^{2}}x+\frac{\dot{a}%
(t)}{a(t)}y\right)  +\left(  \frac{\dot{a}(t)}{a(t)}x-\frac{\xi}{a(t)^{2}%
}y\right)  \frac{\partial}{\partial x}\left(  \frac{\xi}{a(t)^{2}}x+\frac
{\dot{a}(t)}{a(t)}y\right) \\
+\left(  \frac{\xi}{a(t)^{2}}x+\frac{\dot{a}(t)}{a(t)}y\right)  \frac
{\partial}{\partial y}\left(  \frac{\xi}{a(t)^{2}}x+\frac{\dot{a}(t)}%
{a(t)}y\right)  +K\gamma\frac{f(s)^{\gamma-2}}{a(t)^{2\gamma}}2y\dot{f}(s)
\end{array}
\right]
\end{equation}%
\begin{equation}
=\rho\left[
\begin{array}
[c]{c}%
-\frac{2\xi\dot{a}(t)}{a(t)^{3}}x+\left(  \frac{\ddot{a}(t)}{a(t)}-\frac
{\dot{a}(t)^{2}}{a(t)^{2}}\right)  y+\left(  \frac{\dot{a}(t)}{a(t)}%
x-\frac{\xi}{a(t)^{2}}y\right)  \frac{\xi}{a(t)^{2}}\\
+\left(  \frac{\xi}{a(t)^{2}}x+\frac{\dot{a}(t)}{a(t)}y\right)  \frac{\dot
{a}(t)}{a(t)}+K\gamma\frac{f(s)^{\gamma-2}}{a(t)^{2\gamma}}2y\dot{f}(s)
\end{array}
\right]
\end{equation}%
\begin{equation}
=\rho\left[  \left(  \frac{\ddot{a}(t)}{a(t)}-\frac{\xi^{2}}{a(t)^{4}}\right)
y+K\gamma\frac{f(s)^{\gamma-2}}{a(t)^{2\gamma}}2y\dot{f}(s)\right]
\end{equation}%
\begin{equation}
=\frac{\rho y}{a(t)^{2\gamma}}\left[  \lambda+2K\gamma f(s)^{\gamma-2}\dot
{f}(s)\right]
\end{equation}%
\begin{equation}
=0.
\end{equation}
with the Emden equation (\ref{Endemeqeq2-Drotation}) and function
(\ref{fexact}).

Therefore, for $\xi\neq0$, we have the vortical and self-similar flows for the
compressible Euler equations (\ref{e1}) in 2D \cite{Makino93exactsolutions}%
.\newline To ensure the non-negative density function, we require that%
\begin{equation}
\rho(t,\vec{x})=\frac{\max\left(  \left(  -\frac{\lambda(\gamma-1)}{2K\gamma
}s+\alpha\right)  ^{\frac{1}{\gamma-1}},\text{ }0\right)  }{a(t)^{2}}.
\end{equation}

This completes the proof.
\end{proof}

\begin{lemma}
For the Emden equation,
\begin{equation}
\left\{
\begin{array}
[c]{c}%
\ddot{a}(t)-\frac{\xi^{2}}{a(t)^{3}}-\frac{\lambda}{a(t)^{2\gamma-1}}=0\\
a(0)=a_{0}>0,\dot{a}(0)=a_{1},
\end{array}
\right.  \label{ModifiedEmden}%
\end{equation}
with arbitrary constants $\xi\neq0$, $\lambda$ and $\gamma>1$, the total
energy is%
\begin{equation}
E(t):=\frac{\dot{a}(t)^{2}}{2}+\frac{\xi^{2}}{2a(t)^{2}}+\frac{\lambda
}{(2\gamma-2)a(t)^{2\gamma-2}}.
\end{equation}
We have the following.\newline(1) For $1<\gamma<2$, if $E(0)<0$, the solution
is time-periodic; otherwise, the solution is global.\newline(2) For
$\gamma=2,$\newline(2a) with $\xi^{2}\geq-\lambda$, the solution is
global;\newline(2b) with $\xi^{2}<-\lambda,$ the solution blows up at a finite
time if%
\begin{equation}
a_{1}<\frac{\sqrt{-\lambda-\xi^{2}}}{a_{0}};
\end{equation}
otherwise, the solution is global.\newline(3) For $\gamma>2,$\newline(3a) with
$\lambda\geq0$, the solution is global;\newline(3b) with $\lambda<0$ and a
constant $a_{Max}=\left(  \frac{-\lambda}{\xi^{2}}\right)  ^{\frac{1}%
{2\gamma-4}}$,\newline(3bI) and $a_{0}\geq a_{Max}$, \newline if $E(0)\leq
F_{pot}(a_{Max})$ or $E(0)>F_{pot}(a_{Max})$ with $a_{1}\geq0$, the solution
is global; otherwise, the solution blows up at a finite time.\newline(3bII)
and $a_{0}<a_{Max}$, \newline if $E(0)\geq F_{pot}(a_{Max})$ with $a_{1}>0$,
the solution is global; otherwise, the solution blows up at a finite time.
\end{lemma}

\begin{proof}
For equation (\ref{ModifiedEmden}), we multiply $\dot{a}(t)$ and then
integrate it, as follows:%
\begin{equation}
\frac{\dot{a}(t)^{2}}{2}+\frac{\xi^{2}}{2a(t)^{2}}+\frac{\lambda}%
{(2\gamma-2)a(t)^{2\gamma-2}}=E(t),\label{rt1}%
\end{equation}
with a constant $E(0)=\frac{a_{1}^{2}}{2}+\frac{\xi^{2}}{2a_{0}^{2}}%
+\frac{\lambda}{(2\gamma-2)a_{0}^{2\gamma-2}}$.\newline We define the kinetic
energy as:
\begin{equation}
F_{kin}:=\frac{\dot{a}(t)^{2}}{2},
\end{equation}
and the potential energy as:%
\begin{equation}
F_{pot}=\frac{\xi^{2}}{2a(t)^{2}}+\frac{\lambda}{(2\gamma-2)a(t)^{2\gamma-2}%
}.\label{potential}%
\end{equation}
The total energy is conserved thusly:%
\begin{equation}
\frac{dE(t)}{dt}=\frac{d}{dt}(F_{kin}+F_{pot})=0.
\end{equation}
The classical energy method for second-order autonomous ordinary differential
equations (which readers can refer to pages 793--798 in \cite{NDE}), can be
applied to analyze the corresponding qualitative properties of the Emden
equation (\ref{ModifiedEmden}).\newline(1) For $1<\gamma<2$, there exists a
unique global minimum for the potential function (\ref{potential}) for
$a(t)>0,$ and $\underset{a(t)\rightarrow0^{+}}{\lim}F_{pot}(a(t))=+\infty$ and
$\underset{a(t)\rightarrow+\infty}{\lim}F_{pot}(a(t))=0$. If%
\begin{equation}
E(0)=\frac{a_{1}^{2}}{2}+\frac{\xi^{2}}{2a_{0}^{2}}+\frac{\lambda}%
{(2\gamma-2)a_{0}^{2\gamma-2}}<0,
\end{equation}
the solution is time-periodic; otherwise, it is global.$\newline$(2) For
$\gamma=2,$ the ordinary differential equation (\ref{ModifiedEmden})
degenerates into the classical Emden equation%
\begin{equation}
\ddot{a}(t)=\frac{\xi^{2}+\lambda}{a(t)^{3}}.
\end{equation}
(2a) For $\xi^{2}\geq-\lambda$, the solution is global.\newline(2b) For
$\xi^{2}<-\lambda,$ the potential function (\ref{potential}) is strictly
increasing with $\underset{a(t)\rightarrow+\infty}{\lim}F_{pot}(a(t))=0$. The
solution blows up at a finite time, if%
\begin{equation}
a_{1}<\frac{\sqrt{-\lambda-\xi^{2}}}{a_{0}};
\end{equation}
otherwise, the solution is global.$\newline$(3) For $\gamma>2,$\newline(3a)
with $\lambda\geq0$, the potential function (\ref{potential}) is a decreasing
function of $a(t)>0$. Thus, the solution is global.\newline(3b) With
$\lambda<0$, the potential function (\ref{potential}) achieves a unique global
maximum, $F_{pot}(a_{Max})$, with $a_{Max}=\left(  \frac{-\lambda}{\xi^{2}%
}\right)  ^{\frac{1}{2\gamma-4}}$ for $a(t)>0$. $\newline$(3bI) And $a_{0}\geq
a_{Max}$, \newline if $E(0)\leq F_{pot}(a_{Max})$ or $E(0)>F_{pot}(a_{Max})$
with $a_{1}\geq0$, the solution is global; otherwise, the solution blows up at
a finite time.\newline(3bII) and $a_{0}<a_{Max}$, \newline if $E(0)\geq
F_{pot}(a_{Max})$ with $a_{1}>0$, the solution is global; otherwise the
solution blows up at a finite time.

The proof is now complete.
\end{proof}

The foregoing lemma makes it easy to determine the blowup or global existence
of the corresponding solutions (\ref{2-Dg>1Rotation}) for the compressible
Euler equations (\ref{e1}) in 2D.

\begin{corollary}
The total energy is defined%
\begin{equation}
E(t):=\frac{\dot{a}(t)^{2}}{2}+\frac{\xi^{2}}{2a(t)^{2}}+\frac{\lambda
}{(2\gamma-2)a(t)^{2\gamma-2}},
\end{equation}
for the Emden equation (\ref{2-Dg>1Rotation})$_{3}$.\newline(1) For
$1<\gamma<2$, if $E(0)<0$, solution (\ref{2-Dg>1Rotation}) is time-periodic;
otherwise, solution (\ref{2-Dg>1Rotation}) is global.\newline(2) For
$\gamma=2,$\newline(2aI) with $\xi^{2}>-\lambda$ or $\xi^{2}=-\lambda$ and
$a_{1}\geq0$, solution (\ref{2-Dg>1Rotation}) is global;\newline(2aII) with
$\xi^{2}=-\lambda$ and $a_{1}<0$, solution (\ref{2-Dg>1Rotation}) blows up at
$T=-a_{1}/a_{0}$.\newline(2b) with $\xi^{2}<-\lambda,$ solution
(\ref{2-Dg>1Rotation}) blows up at a finite time if%
\begin{equation}
a_{1}<\frac{\sqrt{-\lambda-\xi^{2}}}{a_{0}};
\end{equation}
otherwise, solution (\ref{2-Dg>1Rotation}) is global.\newline(3) For
$\gamma>2,$\newline(3a) with $\lambda\geq0$, solution (\ref{2-Dg>1Rotation})
is global;\newline(3b) with $\lambda<0$ and a constant $a_{Max}=\left(
\frac{-\lambda}{\xi^{2}}\right)  ^{\frac{1}{2\gamma-4}}$,\newline(3bI) and
$a_{0}\geq a_{Max}$, \newline if $E(0)\leq F_{pot}(a_{Max})$ or $E(0)>F_{pot}%
(a_{Max})$ with $a_{1}\geq0$, solution (\ref{2-Dg>1Rotation}) is global;
otherwise, solution (\ref{2-Dg>1Rotation}) blows up at a finite time.\newline%
(3bII) and $a_{0}<a_{Max}$, \newline if $E(0)\geq F_{pot}(a_{Max})$ with
$a_{1}>0$, solution (\ref{2-Dg>1Rotation}) is global; otherwise, solution
(\ref{2-Dg>1Rotation}) blows up at a finite time.
\end{corollary}

\section{Conclusion and Discussion}

This paper provides a class of self-similar vortical flows for the 2D
compressible Euler equations. The result presented herein complements Makino's
solutions in radial symmetry without rotation. In addition, the corresponding
blowup or global existence conditions are classified by analyzing the new
Emden equation (\ref{2-Dg>1Rotation})$_{3}$. Solution (\ref{2-Dg>1Rotation})
is also the solution of the compressible Navier-Stokes equations in 2D:
\begin{equation}
\left\{
\begin{array}
[c]{rl}%
{\normalsize \rho}_{t}{\normalsize +\nabla\cdot(\rho\vec{u})} &
{\normalsize =}{\normalsize 0}\\[0.08in]%
\rho\lbrack\vec{u}_{t}+(\vec{u}\cdot\nabla)\vec{u}]+K\nabla\rho^{\gamma} &
=\mu\Delta\vec{u},
\end{array}
\right.
\end{equation}
with a positive constant $\mu$.

Based on the existence of the rotational and self-similar solutions
(\ref{2-Dg>1Rotation}) in 2D, it is natural to conjecture that there exists a
class of rotational solutions for the compressible Euler equations (\ref{e1})
in 3D, that complements those in radial symmetry,
\cite{Makino93exactsolutions}, \cite{YuenCNSNS2012}:\emph{
\begin{equation}
\left\{
\begin{array}
[c]{l}%
\rho=\frac{f(s)}{\underset{k=1}{\overset{3}{\Pi}}a_{k}}\\
u_{i}=\frac{\dot{a}_{i}}{a_{i}}\left(  x_{i}-d_{i}^{\ast}\right)  +\dot{d}%
_{i}^{\ast}\text{, \ \ \ \ \ for }i=1,2,3,
\end{array}
\right. \label{e10}%
\end{equation}
where
\begin{equation}
f(s)=\max\left(  \left(  -\frac{\xi(\gamma-1)}{2K\gamma}s+\alpha\right)
^{\frac{1}{\gamma-1}},\text{ }0\right)  ,\label{e11}%
\end{equation}
with
\begin{equation}
d_{i}^{\ast}=d_{i0}^{\ast}+td_{i1}^{\ast}\;,\quad\quad\quad s=\underset
{k=1}{\overset{3}{\sum}}(\frac{x_{k}-d_{k}^{\ast}}{a_{k}})^{2}.\label{e12}%
\end{equation}
Here, $\xi,d_{i0}^{\ast},d_{i1}^{\ast}$, $\alpha\geq0$ are constants, and
functions $a_{i}=a_{i}(t)$:
\begin{equation}
\left\{
\begin{array}
[c]{l}%
\ddot{a}_{i}=\frac{\xi}{a_{i}\left(  \underset{k=1}{\overset{3}{\Pi}}%
a_{k}\right)  ^{\gamma-1}}\text{ \ \ \ \ \ \ for }i=1,2,3\\[-0.02in]%
a_{i}(0)=a_{i0}>0,\text{ }\dot{a}_{i}(0)=a_{i1},
\end{array}
\right. \label{e13}%
\end{equation}
with arbitrary constants $a_{i0}$ and $a_{i1}$}.\newline Further research will
be carried out to investigate the possibility of these vortical solutions in 3D.


\begin{thebibliography}{99}                                                                                               %


\bibitem {CIP}Cercignani C., Illner R., and Pulvirenti M. (1994), The
Mathematical Theory of Dilute Gases, Applied Mathematical Sciences
\textbf{106}, Springer-Verlag, New York.

\bibitem {CW}Chen G.Q. and Wang D.H. (2002), The Cauchy Problem for the Euler
Equations for Compressible Fluids, Handbook of Mathematical Fluid Dynamics,
\textbf{I}, 421--543, North-Holland, Amsterdam.

\bibitem {ConstantinA}Constantin A. (2006), Breaking Water Waves,
\textit{Encyclopedia of Mathematical Physics}, 383--386, Elsevier.

\bibitem {Einzel}Einzel D. (2006), Superfluids, \textit{Encyclopedia of
Mathematical Physics}, 115--121, Elsevier.

\bibitem {KYIso}Kwong M.K. and Yuen M.W., Periodical Self-similar Solutions to
the Isothermal Euler-Poisson Equations for Spiral Galaxies, Pre-print.

\bibitem {Li}Li T.H. (2005), Some Special Solutions of the Multidimensional
Euler Equations in $R^{N}$, \textit{Comm. Pure Appl. Anal.} \textbf{4}, 757--762.

\bibitem {LW}Li T.H. and Wang D.H. (2006), Blowup Phenomena of Solutions to
the Euler equations for Compressible Fluid Flow,\textit{ J. Differential
Equations} \textbf{221}, 91--101.

\bibitem {Lion}Lions P.L. (1996, 1998), Mathematical Topics in Fluid
Mechanics, \textbf{1}, \textbf{2}, \textit{Oxford Lecture Series in
Mathematics and its Applications}, \textbf{3}, \textbf{10}, Oxford Science
Publications. The Clarendon Press, Oxford University Press, New York.

\bibitem {Makino93exactsolutions}Makino T. (1993), Exact Solutions for the
Compressible Eluer Equation, \textit{Journal of Osaka Sangyo University
Natural Sciences} \textbf{95}, 21--35.

\bibitem {NDE}Nagle R.K., Saff E.B.and Snider A.D. (2008), Fundamentals of
Differential Equations and Boundary Value Problems, 5th edition,
Addison-Wesley, Reading.

\bibitem {YuenJMAA2008a}Yuen M.W. (2008), Analytical Blowup Solutions to the
2-dimensional Isothermal Euler-Poisson Equations of Gaseous Stars, \textit{J.
Math. Anal. Appl.} \textbf{341, }445--456.

\bibitem {Yuen3DexactEuler}Yuen M.W. (2011), Exact, Rotational, Infinite
Energy, Blowup Solutions to the 3-dimensional Euler Equations, \textit{Phys.
Lett. A} \textbf{375}, 3107--3113.

\bibitem {YuenPLA}Yuen M.W. (2011), Perturbational Blowup Solutions to the
Compressible 1-dimensional Euler Equations, \textit{Phys. Lett. A
}\textbf{375}, 3821--3825.

\bibitem {YuenCNSNS2012}Yuen M.W. (2012),\textit{ }Self-similar Solutions with
Elliptic Symmetry for the Compressible Euler and Navier--Stokes Equations in
$R^{N}$, \textit{Commun. Nonlinear Sci. Numer. Simul.} \textbf{17},
4524--4528.\newline\newline An, H.L. and Yuen, M.W. (2013), Supplement to
"Self-similar Solutions with Elliptic Symmetry for the Compressible Euler and
Navier-Stokes Equations in $R^{N}$" [Commun Nonlinear Sci Numer Simu. 17
(2012) 4524-4528], \textit{Commun. Nonlinear Sci. Numer. Simul.} \textbf{18}, 1558--1561.

\bibitem {ZZ}Zhang T. and Zheng Y.X. (1997), Exact Spiral Solutions of the
Two-dimensional Euler Equations, \textit{Discrete Contin. Dynam. Systems}
\textbf{3}, 117--133.
\end{thebibliography}
\end{document}